\documentclass[10pt]{amsart}
\usepackage{amsfonts}
\usepackage{amssymb}
\usepackage{amsmath,amsthm}
\usepackage{amscd}
\usepackage{graphics}
\usepackage{pdfsync}
\usepackage{mathrsfs}
\usepackage{enumitem}
\usepackage{stmaryrd}
\usepackage[colorlinks,linkcolor=blue,citecolor=blue,urlcolor=blue]{hyperref}

\newcommand{\beq}{\begin{equation}}
\newcommand{\eeq}{\end{equation}}
\newcommand{\bea}{\begin{eqnarray}}
\newcommand{\eea}{\end{eqnarray}}
\newcommand{\nn}{\nonumber}
\newcommand{\noi}{\noindent}

\newtheorem{definition}{Definition}
\newtheorem{proposition}{Proposition}
\newtheorem{theorem}{Theorem}

\theoremstyle{definition}

\newtheorem{remark}{\textbf{Remark}}

\begin{document}


\title{Universality Classes and Information-Theoretic Measures of Complexity via Group Entropies}

\author{Piergiulio Tempesta}
\address{Instituto de Ciencias Matem\'aticas, C/ Nicol\'as Cabrera, No 13--15, 28049 Madrid, Spain\\ and Departamento de F\'{\i}sica Te\'{o}rica, Facultad de Ciencias F\'{\i}sicas, Universidad
Complutense de Madrid, 28040 -- Madrid, Spain }
\author{Henrik Jeldtoft Jensen}
\address{Centre for Complexity Science and Department of Mathematics, Imperial College London, South Kensington Campus, SW7 2AZ, UK
and \\ Institute of Innovative Research, Tokyo Institute of Technology, 4259, Nagatsuta-cho, Yokohama 226-8502, Japan}
\email{piergiulio.tempesta@icmat.es, ptempest@ucm.es}
\date{October 17, 2019}
\maketitle


\begin{abstract}
We introduce a class of information measures based on group entropies, allowing us to describe the information-theoretical properties of complex systems. These entropic measures are nonadditive, and
are mathematically deduced from a series of natural requirements. In particular, we introduce an extensivity postulate as a natural requirement for an information measure to be meaningful. The information measures proposed are suitably defined for describing universality classes of complex systems, each characterized by a specific phase space growth rate function.
\end{abstract}


\tableofcontents

\section{Introduction}
\subsection{A new perspective on complexity}
The aim of this paper is to propose a general theoretical construction that allows us to associate a given class of complex systems with a suitable information measure adapted to this class, and expressed by an entropic functional mathematically deduced from a set of axioms, belonging to the family of group entropies (\cite{PT2011PRE}, \cite{PT2016AOP}, \cite{PT2016PRA}).

The main idea behind our approach is simple. In a broad range of applications, including physical and social sciences, economics and neurosciences, it is customary to use information measures based on the additive Shannon entropy (and its quantum version, the von Neumann entropy). Standard and useful indicators of complexity commonly adopted in the literature are indeed the mutual information or the relative entropy.

However, instead of using an information entropic measure defined a priori, and based on a (sometimes not fully justified) assumption about additivity, one may proceed differently. We propose to look for new information measures, written in terms of entropic functionals that respect the specific properties of the system, or family of systems, under consideration.


To this aim, we shall prove a theorem that allows to associate with a given universality class of systems a specific entropic measure, constructed in a completely algorithmic way. This measure will be extensive and non-additive, and will depend explicitly on the phase space growth rate function which characterizes the universality class considered. From a mathematical point of view, the derivation of these entropic measure is a direct consequence of an axiomatic approach, based on formal group theory.
Using the group-theoretical entropic measures so defined, we shall construct a new family of information-theoretical measures of complexity.

The deep insights represented by the Tononi-Edelman-Spons Integrated Information concept is traditionally formulated mathematically in terms of sums of conditioned entropies of partitions of the considered system in particular the brain. However this mathematical representation does suffer from limitations \cite{MSB2018}. The group theoretic entropies introduced in the present paper offers an alternative mathematical implementation of the original TES idea, without need of introducing conditioning. We explain below how a new complexity measure based on the group entropies offers a way to characterize the degree of entanglement of brain dynamics and, moreover gives a way to compute the capacity of a neuronal network of a certain size. In Sec. 5.2 we formulate this as a precise mathematical conjecture.

\subsection{A group-theoretical approach to information theory: Group entropies}
Since the work of Boltzmann, perhaps the most relevant problem of statistical mechanics has been the study of the connections between the statistical properties of a complex system at a microscopic level,
and the corresponding macroscopic thermodynamic properties.

The probabilistic point of view of Boltzmann, further developed by Gibbs, Planck and many others, was questioned from the very beginning by Einstein. As is well known, Einstein argued that
probabilities must follow dynamics, and not vice versa. In other words, the frequency of occupation of the different regions of phase space should not be given a priori,
but determined from the equations of motions. A conciliation between these two points of view is still an unsolved problem: as surmised by Cohen  
\cite{Cohen2002}, a combination of
statistics and dynamics is perhaps a necessary way out to describe the statistical mechanics of a complex system.

In this perspective, it is quite natural to hypothesize that the geometry of the phase space associated with a given complex system  plays a crucial role in the characterization of its main  information-theoretical, dynamical and  statistical features.

In this paper we shall try to shed new light on this aspect.
We shall adopt in some sense an ``intermediate" point of view: Indeed, instead of focusing on the dynamics of a specific system, we can consider \textit{universality classes of systems},
defined in the following way.

Let us denote by $W=W(N)$ the phase space growth rate of a given system, i.e. the function describing asymptotically the number of allowed microstates as a function of the number of particles $N$.\footnote{Since we are interested essentially in the large $N$ limit, we can always think of $W(N)$ as an integer number (i.e. we shall identify it with its integer part) for any choice of $W$.}
A universality class of systems is defined to be the set of all possible systems sharing the same growth rate function $W=W(N)$.
For instance, many physical systems will be associated with an occupation law of the form $W(N)=k^{N}$, $k\in\mathbb{N}\backslash\{0\}$.
Other natural choices are, for instance, $W(N)= N^{\alpha}$ and $W(N)= N!$. Generally speaking, we can partition all possible universality classes into three families: the \textit{subexponential}, the \textit{exponential} or the \textit{super-exponential family} depending on whether $W(N)< e^{N}$,  $W(N)= e^{N}$, $W(N)> e^{N}$ for large $N$, respectively.

We will show that by means of a group-theoretical approach, given any universality class one can construct in a purely deductive and axiomatic way an entropic functional representing a suitable information measure for that class.

This approach is clearly inspired by the research on generalized entropies that in the last few decades captured a considerable interest. In particular, we will use the notion of \textit{group entropies}, introduced in \cite{PT2011PRE}, and settled in general terms in \cite{PT2016AOP}, \cite{PT2016PRA} (see also \cite{JT2018ENT} for a recent review). Essentially, a group entropy is a generalized entropy that has associated a group law, which describes how to compose the entropy when we merge two independent systems into a new one.

Said more formally, the physical origin of the group theoretical structure relies on the axiomatic formulation of the notion of entropy due to Shannon and Khinchin.
The first three Shannon-Khinchin axioms \cite{Shannon}, \cite{Shannon2} \cite{Khinchin} represent fundamental, non-negotiable requirements that an entropy $S[p]$ should satisfy to be physically meaningful. Essentially, they amount to the following properties: \\
\noi (SK1) $S[p]$ is continuous with respect to all variables $p_1,\ldots,p_W$. \\
\noi (SK2) $S[p]$ takes its maximum value over the uniform distribution (it implies concavity). \\
\noi (SK3) $S[p]$ is expansible: adding an event of zero probability does not affect the value of $S[p]$.

However, although necessary, these properties are still not sufficient for thermodynamical purposes. Indeed, we need another crucial ingredient: \textit{composability} \cite{Tbook}. In \cite{PT2016AOP}, \cite{PT2016PRA}, this property has been reformulated and related to group theory as follows.


\subsection{Composability axiom} An entropy is said to be \textit{composable} if there exists a smooth function $\Phi(x,y)$ such that, given two statistically independent subsystems $A$ and $B$ of a complex system, $S(A\cup B)=\Phi(S(A), S(B))$, when the two subsystems are defined over \textit{any arbitrary probability distribution} of $\mathcal{P}_{W}$. In addition, we shall require that: \\
\noi (C1) Symmetry: $\Phi(x,y)=\Phi(y,x)$. \\
\noi (C2) Associativity: $\Phi(x,\Phi(y,z))=\Phi(\Phi(x,y),z)$. \\
\noi (C3) Null-Composability: $\Phi(x,0)=x$. \\

Observe that the requirements $(C1)$--$(C3)$ are fundamental ones: they impose the independence of the composition process with respect to the order of $A$ and $B$, the possibility of composing three independent subsystems in an arbitrary way, and the requirement that, when composing a system with another one having zero entropy, the total entropy remains unchanged. In our opinion, these properties are also ``non-negotiable": indeed, no thermodynamics would be easily conceivable without these properties. From a mathematical point of view, the properties above define a \textit{group law}. In this respect, the theory of formal groups \cite{Boch}, \cite{Haze}, \cite{Serre} offers a natural language in order to formulate the theory of generalized entropies.
Notice that the above construction define a full group, since the existence of a power series $\phi(x)$ such that $\Phi(x,\phi(x))=0$ (i.e. the ``inverse") is a consequence of the previous axioms \cite{BMN, Haze}.
Let $\{p_i\}_{i=1,\cdots,W}$, $W\geq 1$, with $\sum_{i=1}^{W}p_i=1$, be a discrete probability distribution; the set of all discrete probability distributions with $W$ entries will be denoted by $\mathcal{P}_{W}$.
\begin{definition}
A group entropy is a function $S: P_{W}\in \mathbb{R}^{+} \cup \{0 \}$ which satisfies the axioms (SK1)-(SK3) and the composability axiom.
\end{definition}
For recent applications of the notion of group entropy in Information Geometry and the theory of divergences \cite{Amari2016}, see e.g. \cite{RRT2019}.
\section{The extensivity postulate}
Our approach is crucially based on the following extensivity postulate
\par
\textbf{Postulate [EP]}. \textit{Given a system in its most disordered state (the uniform distribution), the amount of its disorder increases proportionally to the number $N$ of its constituents}.
\par
In other words, we shall require that, if $S$ is an information measure of order/disorder for that system, we must have $S(N)/N =const$. A weaker condition, suitable for macroscopic systems, is the asymptotic condition
\beq \label{ext}
\lim_{N\to\infty} \frac{S(N)}{N}=const.
\eeq
We stress that in this paper we are not considering thermodynamics, but a purely information-theoretical context.
In order to satisfy the previous postulate, we shall construct \textit{entropic} information measures, based on group entropies, namely a class of generalized entropies constructed axiomatically from group
theory.

For entropic functionals, the postulate EP  corresponds to requiring \textit{extensivity}.  Indeed, a fundamental requirement, pointed out already by Clausius, is that thermodynamic entropy, \textit{as a function of $N$}, must grow linearly in $N$ in the thermodynamic limit when $N\to \infty$. It is immediately seen from the classical relation $S=k_{\mathcal{B}} \ln W$, valid in the case of equal probabilities, that Boltzmann's entropy is extensive for the universality class $W(N)\sim k^{N}$,which typically contains ergodic systems.  However, the Shannon entropy $S=-\sum_{i=1}^{W} p_i \ln p_i$ is not extensive over other universality classes. Therefore, in order to satisfy the postulate EP, it appears clear that new entropic functionals should be found.



A crucial problem emerges naturally: given a complex system, how can one associate to it a meaningful information measure? This is the main question we address in this paper. We will prove that, surprisingly, there is a possible answer, simple and deductive.

Our main result is indeed the following: \textit{For any universality class of systems there exists an information measure satisfying the axioms (SK1)--(SK3), the composability axiom and postulate EP}.


Consequently, we shall propose a deductive construction of a group entropy associated with a given universality class. At the heart of this construction, there is a very simple idea: an admissible entropic information measure has associated an intrinsic group-theoretical structure, responsible of essentially all the properties of the considered entropy. This structure, provided by a specific group law, comes from the idea of allowing the composition of statistically independent systems in a thermodynamically meaningful way.

\section{A dual construction of entropies}
 Let $G\left( t\right) =\sum_{k=1}^{\infty} a_k \frac{t^{k}}{k}$ be a real analytic function, where $\{a_{k}\}_{k\in\mathbb{N}}$ a real sequence, with $a_{1}=1$, such that the function  $S_{U}:\mathcal{P}_{W}\rightarrow \mathbb{R}^{+}\cup\{0\}$, 
defined by
\beq
S_{U}(p_1,\ldots,p_W):=\sum_{i=1}^{W}p_i \hspace{1mm} G\left(\ln \frac{1}{p_i}\right), \label{entropy}
\eeq
is a concave one. This function is the \textit{universal-group entropy}.
The two most known examples of entropies of this class are the Boltzmann-Gibbs entropy $S_{BG}[p]=\sum_{i=1}^{W}p_i \ln \left(\frac{1}{p_i}\right)$ and the Tsallis entropy $S_{q}=\frac{1-\sum_{i=1}^{W}p_{i}^{q}}{q-1}$ \cite{Tsallis1}. Essentially all the entropic functionals known in the literature are directly related with the class \eqref{entropy}.
The composability axiom widely generalizes the fourth Shannon-Khinchin axiom. Needless to say, when $\Phi(x,y)=x+y$, we get back the original version of the axiom (SK4), which states the additivity of the Boltzmann entropy with respect to the composition of two statistically independent subsystems. For $\Phi(x,y)=x+y+(1-q) xy$, we have the composition law of $S_{q}$ entropy, and so on.

The composability property, if required for any choice of the probability distributions allowed to $A$ and $B$, is a nontrivial one. A theorem proved in \cite{ET2017JSTAT} states indeed that in the class of trace-form entropies only the entropies $S_{BG}$ and $S_{q}$ are composable (uniqueness theorem).

However, the remaining cases can be at most \textit{weakly composable}: the group law $\Phi(x,y)$ is defined at least over the uniform distribution. This case is especially important for thermodynamics, but it is not sufficient to cover many other physical situations.


The above discussion motivates the study of a family of entropies which are not in the trace-form class. In \cite{PT2016PRA}, the family of $Z$-entropies has been introduced. They generalize both the Boltzmann and the R\'enyi entropies and are \textit{strongly composable}.

An important result, due to Lazard, states that there exists a power series of the form $G(t)=t+a_2/ 2 t^2+\ldots$ such that, given a smooth function $\Phi(x,y)$
that satisfies the properties  $(C1)$-$(C3)$, it can be represented in the form (See SM)
\beq
\Phi(x,y)=G\left(G^{-1}(x)+G^{-1}(y)\right). \label{Phi}
\eeq

The general form of the entropies of this class is given by
\beq
Z_{G,\alpha}(p_1,\ldots,p_W):=\frac{\ln_{G} \left(\sum_{i=1}^{W} p_{i}^{\alpha}\right)}{1-\alpha}, \label{Zent}
\eeq
where $\ln_{G}(x)$ is a generalized logarithm  and $\alpha$ is a real parameter. When $0<\alpha<1$, the $Z_{G,\alpha}$ entropy is concave; when $\alpha>0$, is Schur-concave. Precisely, a generalized group logarithm is a continuous, concave, monotonically increasing function $\ln_{G}: \mathbb{R} \to \mathbb{R}$, possibly depending on a set of real parameters, satisfying a functional equation (a group law). It can be considered to be a deformation of the standard Neperian logarithm (See SM).

 For the purposes this paper, the concavity requirement will not be necessary. Throughout this paper, we shall put the Boltzmann constant $k_{\mathcal{B}}=1$.

\section{The main reconstruction theorem: From phase space to group entropies}

\subsection{Main result}
The following theorem formalizes in a rigorous way, completes and unifies several preliminary ideas already expressed in a heuristic way in previous papers of ours \cite{PT2016PRA}, \cite{JT2018ENT}. However the main result is new: a compact expression of an entropy directly expressed as a function of the phase space growth rate. Here we shall replace $N$ with a continuous variable interpolating the discrete values of $N$. Also, we shall introduce a sufficiently regular function $\mathcal{W}=\mathcal{W}(x)$, that we shall call a \textit{growth function}. For the purposes of this article, from now on we shall require that $\mathcal{W}$ is at least a  monotonic, strictly increasing function of class $C^{1}(\mathbb{R}^{+})$.  Hereafter, $W$ will denote the integer part of $\mathcal{W}$\footnote{Usually, in the literature $W(N)$ and $\mathcal{W}(N)$ are identified for notational simplicity. For large values of $N$, the discrepancy between the two values is numerically very small}.

\begin{theorem}\label{theo1}
Let $\mathcal{W}$ be a phase space growth function, corresponding to a given universality class of statistical systems. Then there exists a unique entropy in the $Z$-class which satisfies the extensivity postulate. This entropy is given by
\beq \label{maineq}
Z_{G,\alpha}(p_1,\ldots,p_{W})= \frac{1}{(\mathcal{W}^{-1})'(1)}\left(\mathcal{W}^{-1}\left(\bigg(\sum_{i=1}^{W}p_i^{\alpha}\bigg)^{\frac{1}{1-\alpha}}\right)- \mathcal{W}^{-1}(1) \right)\ ,
\eeq
with $\alpha>0$ and it is assumed $(\mathcal{W}^{-1})'(1)\neq 0$.
\end{theorem}
\textit{Proof}.
Let us assume that the asymptotic behaviour of $Z_{G,\alpha}(p)$, for large values of $N$, is given by
\beq \label{ZkN}
Z_{G,\alpha}[\mathcal{W}(N)]=k N +k_{0}
\eeq
for suitable constants $k$ and $k_{0}$ (which a priori could depend on thermodynamic variables, but not on $N$). This condition obviously implies extensivity, namely eq. \eqref{ext}. The more general form \eqref{ZkN}, which also includes  the constant $k_{0}$ is introduced for further convenience.
\par
Any generalized logarithm can be represented in the form (see SM)
\beq
\ln_{G}(x)=G(\ln x),
\eeq
for an invertible function $G$, whose behaviour around zero is given by
\beq \label{eq:Garound0}
G(t)= t+ O(t^2)
\eeq
which also implies the property
\beq\label{G0}
G(0)=0
\eeq
This function is the ``group exponential'' that defines the composition law of the entropy by means of relation \eqref{Phi}. Therefore from eqs.  \eqref{Zent}, \eqref{ZkN}, we obtain
\beq \label{eq:Gamma}
G\left(\ln \mathcal{W}(N)^{1-\alpha}\right)=(1-\alpha) k N + \Gamma \ ,
\eeq
where $\Gamma= (1-\alpha)k_{0}$ is another constant. Hence the relation between the phase space growth rate and the group exponential is explicitly given by
\beq
\mathcal{W}(N)=\left(e^{G^{-1}\left[(1-\alpha) k N + \Gamma\right]}\right)^{\frac{1}{^{1-\alpha}}} \ .
\eeq
Let
\[
t:= \ln \mathcal{W}(N)^{1-\alpha} \Longleftrightarrow N=\mathcal{W}^{-1}(e^{\frac{t}{1-\alpha}}) \ .
\]
We obtain from eq. \eqref{eq:Gamma}
\[
G(t)=k (1-\alpha)\mathcal{W}^{-1}(e^{\frac{t}{1-\alpha}})+\Gamma \ .
\]
which implies, using relation \eqref{G0}
\[
G(0)=k \mathcal{W}^{-1}(1) + \Gamma =0 \Longleftrightarrow \Gamma = -k (1-\alpha) \mathcal{W}^{-1}(1) \ ,
\]
Consequently we get
\beq
G(t)=k(1-\alpha)\left(\mathcal{W}^{-1}\left(e^{\frac{t}{1-\alpha}}\right)- \mathcal{W}^{-1}(1)\right).
\eeq
In order to fix the constant $k$, it is sufficient to observe that the condition \eqref{eq:Garound0} implies that
\[
k= \frac{1}{(\mathcal{W}^{-1})'(1)},
\]
where we assume $(\mathcal{W}^{-1})'(1)\neq0$. In other words, the group exponential can be uniquely determined by the specific choice of a universality class of systems.
\par
From the explicit expression of this function we can reconstruct the entropy we are looking for:
\bea \nn
Z_{G,\alpha}[p]= \frac{\ln_{G} \left(\sum_{i=1}^{W} p_{i}^{\alpha}\right)}{1-\alpha}= \frac{1}{(\mathcal{W}^{-1})'(1)}\left(\mathcal{W}^{-1}\left(\bigg(\sum_{i=1}^{W}p_i^{\alpha}\bigg)^{\frac{1}{1-\alpha}}\right)-\mathcal{W}^{-1}(1) \right) \ .
\\ \label{ZW}
\eea

The constant appearing in the r.h.s. of the previous formula guarantees that the entropy vanishes over a certainty state, namely for a distribution where $\exists~i$ such that $p_i=1$, $p_j=0$, $j\neq i$.

Observe that for $\alpha>1$, then $Z_{G,\alpha}[p]$ is still a non-negative function, as it can be ascertained, for instance, by noticing that $G\left(\ln \big(\sum_{i=1}^{W} p_{i}^{\alpha}\big)\right)<0$. 

The case $\alpha\to 1$ is admissible and interesting by itself, and gives us 
\beq
Z_{G,1}[p]= \frac{1}{(\mathcal{W}^{-1})'(1)}\left(\mathcal{W}^{-1}\left(\exp(S_{BG}[p]\right)-\mathcal{W}^{-1}(1) \right)
\eeq
where $S_{BG}[p]$ is the Boltzmann-Gibbs entropy. To conclude the proof, one can ascertain that since $Z_{G,\alpha}[p]$ is the composition of a strictly increasing function with a function which is strictly Schur-concave for $\alpha>0$, then it is still strictly Schur concave in the same interval (see e.g. \cite{MOA2010}, page 89 and \cite{RV1973}). This property is sufficient for satisfying the maximum entropy axiom.
$\Box$
\begin{remark}
The entropy $Z_{G,\alpha}$ can also be expressed up to a multiplicative constant, in the form
\beq
Z_{G,\alpha}(p_1,\ldots,p_{W})= k \left(\mathcal{W}^{-1}\left(\bigg(\sum_{i=1}^{W}p_i^{\alpha}\bigg)^{\frac{1}{1-\alpha}}\right)- \mathcal{W}^{-1}(1) \right)\ ,
\eeq
which does not require the condition $(\mathcal{W}^{-1})'(1)\neq0$. This formulation is consistent with the possibility of changing the units of thermodynamic quantities.
\end{remark}
\begin{remark}
In many applications, the growth function $\mathcal{W}$ is convex in its domain.  Then the resulting $Z_{G,\alpha}[p]$ in particular is concave in a suitable finite interval of values of $\alpha>0$, as in the general construction of \cite{PT2016PRA}. When dealing with  different regularity properties,  for instance in the case of complex systems whose number of degrees of freedom is monotonically decreasing as a function of the number $N$ of particles,  the previous construction should be properly modified.
\end{remark}

An interesting question should be addressed, namely the uniqueness of the previous construction. Certainly, we can derive other group entropies with similar properties. More precisely, a theorem proved in \cite{RRT2019} shows that for a given group-theoretical structure, one can associate a ``tower'' of group entropies sharing the same structure.
At the same time, the $Z$-family defined in eq. \eqref{Zent} is \textit{complete}: as we have just proved, for each universality class there exists a representative in the $Z$-family playing the role of admissible entropy. Also, this representative presents the ``simplest'' functional form within the ``tower'' associated with a given group structure.
Under mild hypotheses, Theorem \ref{theo1} solves completely the problem of determining an entropy suitable for a given universality class. It represent a conceptually and practical powerful tool for constructing infinitely many new entropic functionals (all of them group entropies) tailored for complex systems, emerging from very different contexts: physics, social sciences, etc. Our next step would be to construct an information measure for each of these new entropies. Rather than an ``universal entropy'', valid for any possible complex system, we have, what may be considered more reasonable, a \textit{specific} entropy, unique in the $Z$-class, for each universality class of systems.

\subsection{The group-theoretical structure associated with W(N)}
By construction, the following important property holds.
\begin{proposition}
Let  $\mathcal{W}=\mathcal{W}(x)$ be a phase space growth function.
The corresponding entropy  $Z_{G,\alpha}$ \eqref{maineq} is strictly composable, namely
\[
Z_{G,\alpha}(A\cup B)= \Phi(Z_{G,\alpha}(A), Z_{G,\alpha}(B))
\]
for all complex systems $A$ and $B$, where the group law $\Phi(x,y)$ associated  can be written in terms of $\mathcal{W}$ as
\beq
\phi(x,y)=\lambda \left\{\mathcal{W}^{-1}\Big[
\mathcal{W}\left(\frac{x}{\lambda}+\mathcal{W}^{-1}(1)\right) \mathcal{W}\left(\frac{y}{\lambda}+\mathcal{W}^{-1}(1)\right)
\Big] -\mathcal{W}^{-1}(1)\right\} \ ,
\label{phi_W}
\eeq
where $\lambda= \frac{1}{(\mathcal{W}^{-1})'(1)}$.
\end{proposition}
\begin{proof}
Let us introduce the function
\[
\chi(x):= \frac{1}{1-\alpha}G\big((1-\alpha)x\big) \ .
\]
We observe that
\[
Z_{G,\alpha}(A\cup B)= \frac{G\left(\ln \left(\sum_{i,j} (p_{ij}^{A\cup B})^{\alpha}\right)\right)}{1-\alpha}=  \frac{G\left(\ln \left(\sum_{i} (p_{i}^{A})^{\alpha}\right)\right)+\ln \left(\sum_{j} (p_{j}^{B})^{\alpha}\right)}{1-\alpha}
\]
\[
= \frac{G\left(G^{-1}\left((1-\alpha) Z_{G,\alpha}(p_{i}^{A})\right)+G^{-1}\left((1-\alpha) Z_{G,\alpha}(p_{j}^{B})\right)\right)}{1-\alpha}=\Phi(Z_{G,\alpha}(A), Z_{G,\alpha}(B))
\]
where $\Phi(x,y)=\chi(\chi^{-1}(x)+\chi^{-1}(y))$.
\end{proof}


All the previous discussion can be summarized as follows: Given an universality class of systems whose occupancy law is assigned, we are able to construct an entropic functional which is extensive at the equilibrium, according to the classical principles of thermodynamics, and the requirements of large deviation theory  \cite{Ellis2006}.
Also, the entropy in eq. \eqref{maineq} satisfies the first three SK axioms and is composable, with group law given by the relation \eqref{Phi}.

Though each entropy of the class $Z_{G, \alpha}$ is extensive at the equilibrium in the specific class corresponding to a given phase space growth rate $\mathcal{W}(N)$, in general it may not be so if applied to systems with a different functional dependence of $\mathcal{W}$, i.e. having other occupation laws at the equilibrium.  This is certainly not surprising, since it is also the behaviour of the standard Boltzmann-Gibbs entropy.

\section{Complexity measures from group entropies}

\subsection{Shannon-type integrated measures}

The traditional approach to quantify degrees of interdependence between a number of components in a complex system is based on Shannon's entropy and investigates the difference between combinations of conditioned entropies. An influential example is Tononi's Integrated Information Theory~\cite{Tononi_2012,Balduzzi_Tononi_2008}, which suggests that consciousness can be detected from measures based on Shannon's entropy which by decompositions analyse the relationship if the information stored in the whole and in parts. A related very recent approach seeks to analyse self-organisation of synergetic interdependencies by a focus on joint Shannon entropies~\cite{Rosas_2018}. And finally let us mention the recent Entropic Brain Hypothesis~\cite{CH_2014}, which relates increased consciousness to increase in the Shannon entropy, with less emphasis on how the interdependence of the conscious state is to be characterised.

The group entropies suggest an alternative approach relevant when the number of degrees of freedom $W(N)$ of the entire complex system is different from the Cartesian product of the degrees of freedom of the parts. Intuitively, it appears that faster than exponential growth of $W(N)$ may apply to the brain. Of course we don't know the details of the relationship between the neuronal substrate and the activities of the mind, but at an anecdotal intuitive level it seems that the mind's virtually limitless capacity of deriving and combining associations of associations in grand hierarchical structures is an example of new emergent states added to what can be reached by Cartesian combinations. Say the sate I put my mind into when I try to imagine the mood of a piglet splashing through the waves out somewhere on the deep ocean while the it desperately tries to figure out how it may be able to use its great grand uncles old bagpipe as a means of flotation. Or to take a more frequently occurring mind state, namely the one induced when one listens to Bach's Chaconne from Partita No. 2 in D minor for solo violin. In both cases it seems unlikely that the state of the mind isn't in an emergent state beyond the Cartesian combinations.

Similarly, $W(N)$ may very likely grow much slower that the exponential dependence of Cartesian combination, say when one have a highly restricted system in mind, say, for arguments sake, the financial system under very strong regulations.

In such cases the group entropies offer a way to directly quantify the extend of systemic interdependence without going through part-wise conditioning. The composability axiom relates to how the whole of a system consisting of different parts differs from the system one would obtain by a simple Cartesian combination of the parts. The Shannon entropy can be thought of as directly focusing on the diversity in a system and then, as a second step, address interdependence e.g. by developing various conditioned measures. In contrast the group entropies are directly sensitive both to diversity and to interdependence when the later is sufficiently strong to make $W(N)\neq k^N$.

\subsection{A new indicator of complexity}
Namely, consider the difference
\beq
\Delta(AB) = S(A\times B)-S(AB)= \phi(S(A),S(B)) - S(AB).
\label{Delta}
\eeq
between the entropy  $S(A\times B)$ of the system constructed by Cartesian combination of parts $A$ and $B$and the entropy $S(AB)$ of the entire complex system  containing the fully interacting and interdependent part $A$ and $B$.

This measure can be thought of as a possible generalisation of the usual mutual information and could be useful as an alternative to
Tononi's  Integrated Information  \cite{Tononi_1994}, \cite{Tononi_1998}, \cite{Tononi_2001}, \cite{Balduzzi_Tononi_2008}, \cite{Tononi_2012}  as a measure that can quantify very entangled complex systems such as perhaps consciousness.

One of the attractions of the group theoretic foundation for the entropies is that it allows a consistent procedure that ensures one get the same value for the entropy when dealing with a system $A\times B$ that is, one may indeed say trivially, composed of two sub systems $A$ and $B$ by simply formally considering the combined system $A\times B$  to consist of the Cartesian set of states  $(a,b)$ where $a$ is a state in $A$ and $b$ a state in $B$. Though one may think of this as a mathematical consistency requirement for any entropy, since it should be able to handle such trivial situations, it is in fact much more than that. Namely, no less than the foundation of Boltzmann-Gibbs statistical mechanics and will always be a good approximation, though formally no more than that, in all cases where interdependence between sub-pats of a system can be neglected. Which, we repeat, is of course typically not the case for complex systems.

When long range forces of other kinds of long rage interdependence is at play the system $A\times B$ will be different from the system $AB$ in which all particles from $A$ and from $B$ are allowed to interact, combine and influence in whatever way the situation allows \cite{JPPT2018JPA}.

The complexity measure defined in Eq. (\ref{Delta}) is a way to quantify the extend of the difference between $A\times B$ and $AB$. Since the composition law $\phi(x,y)$ assumes the same functional dependence for trace and non-trace entropies when expressed in terms of $W(N)$ and its inverse, Eq. (\ref{phi_W}) we  conclude that the degree of complexity and its dependence of number of degree of freedom is fully controlled by the functional form of $W(N)$. For simplicity consider the situation where $A=B$ and each subsystem has $N_0$ degrees of freedom. Furthermore, if we restrict ourselves to the uniform ensemble  $p_i=1/W$, where the extensivity property of the group entropies simply make sure that $S[p]=\lambda N$, we have the following expressions for the complexity measure $\Delta(AA)$, which we simply denote $\Delta(N_0)$

\begin{itemize}
\item[(I)] Algebraic -- $W(N)=N^a$:
\beq
\Delta(N_0) = \lambda\bigg(N_0+N_0+\frac{N_{0}^2}{\lambda}\bigg)-\lambda (N_0+N_0)=N_0^2.
\eeq
We might call this the Tsallis case where the the interdependence between particles strongly restrict  the available phase space. The entropy of the Cartesian combination $A\times A$ over shoots the entropy of the fully entangled system $AA$ by $N_0$. One may think of this as indicating that the reduction of phase space involves a restrictive relation between each particle and the $N_0-1$ other particles.

\item[(II)] Exponential -- $W(N)=k^N$:
\beq
\Delta(N_0)=\lambda(N_0+N_0)-\lambda(N_0+N_0)=0.	
\eeq
The Boltzmann-Gibbs case where the entire system effectively is composed of a non-interdependent set of subsystems.
\item[(III)] Super-exponential $W(N)=N^{\gamma N}$:
\begin{align}
\Delta(N_0)&=\lambda\big\{\exp[L(2(1+N_0)\ln(1+N_0))]-1\big\}-\lambda(N_0+N_0)\nonumber\\
&\simeq \lambda(\exp[\ln(2N_0\ln N_0)]-2N_0\nonumber\\
&\simeq 2\lambda(N_0\ln N_0-N_0)\simeq 2\lambda \ln N_0!.	
\end{align}
Here we assumed $N_0\gg 1$ and made use of the fact that $L(x)\simeq \ln x$ asymptotically. \\
The effective dependence of the complexity measure on the factorial suggests  a relation to the super-exponential behaviour of $W(N)$ originating in the creation of new states by forming combinations the particles \cite{JPPT2018JPA}.
\end{itemize}

\section{The complexity of human brain: a conjecture}
Returning to the question of the complexity of neural networks, we point out that in principle the indicator $\Delta(N_0)$ can be used to study the function $W(N_0)$ for the case of human brain. For example, extract probabilities for the various states of the brain measured by brain scans from, say, single neurone potential measurements, fMRI or EEG recordings. The histogram of the simultaneously recorded signals will give us an estimate of the joint probability density $P(x_1,...,x_N)$ where $x_1(t)$ up to $x_N$ are the $N$ recorded times series.  We break the set of time series into two groups, each consisting of $N_0\leq N/2$ time series. Let $P(x_1,...,x_{2N_0})$ characterise the "full" system $AB$ discussed above. We can then extract the distributions for subsystems as the marginalised probabilities.  I.e. let our sub-systems $A$ and $B$ be given by

\[ P_A(x_1,...,x_{N_0}) =\int dx_{N_0+1}\cdots\int d x_N P(x_1,...,x_N)\]
\[P_B(x_{N_0+1},...,x_{2N_0}) =\int dx_1\cdots\int d x_{N_0}dx_{2N_0+1}\cdots dx_N P(x_1,...,x_N),\]
and form the Cartesian non-interacting system $A\times B$ described by
\[P_{AB}(x_1,...,x_{2N_0})=P_A(x_1...,x_{N_0})\times P_B(x_{N_0+1},...,x_{2N_0}). \]
We can now compute $\Delta$ in Eq. (\ref{Delta}) and by varying the number $N_0$ of data streams  included in the two sub-systems, we can check the $N_0$ dependence of $\Delta$. We conjecture that for the  brain, $\Delta$ will depend like $\ln N_0 ! $. This corresponds to case (III) above i.e. corresponding to $W(N)=N^{\gamma N}$. Consequently we can state the following
\vspace{2mm}

\textbf{Conjecture}: The number of brain states typically grows \textit{faster than exponentially} in the number of brain regions involved.

\section{Future perspectives}

The axiomatic approach proposed allows us to associate in a simple way an entropic function with an universality class of systems. In particular, the extensivity axiom selects among the infinitely many group entropies of the $Z$-class a unique functional, which possesses many good properties, all necessary for an information-theoretical interpretation of the functional as an information measure. The standard additivity must be replaced by a more general composability principle that ensures, that in the case of a system composed of statistically independent components, the properties of the compound are consistent with the those of its components.

Many research perspectives are worth being explored in the future.
Generally speaking, our formalism allows for a systematic generalisation of a statistical mechanics description to non-exponential phase spaces.

For instance, we believe that group-theoretical information measures in the study of self-organized criticality could replace Shannon's entropy in several contexts where the number of degrees of freedom grows in a non-exponential way.


We also mention that classifying complex systems without worrying about composability was done in \cite{KHT2018}.

The analysis of time series of data, from this point of view, offer another interesting possibility of testing the present theory.

A quantum version of the present approach, in reference with the study of quantum entanglement for many-body systems represent an important future objective our our research.

Work is in progress along these lines.

\appendix

\vspace{3mm}

\noindent \textbf{SUPPLEMENTARY MATERIAL}.

\vspace{3mm}

\textbf{Groups and entropies: a brief summary.} For sake of completeness, in order to have a self-contained exposition, we shall review here some results concerning formal group theory and its relation with the theory of generalized entropies, following closely refs. \cite{PT2016AOP} and \cite{PT2016PRA}.

%
%
%
%
%

\section{Formal group laws}

We will start by recalling some basic facts and definitions of the theory of formal groups (see \cite{Haze} for a thorough exposition, and \cite{BMN}, \cite{Serre} for a shorter introduction).

Let $R$ be a commutative associative ring  with identity, and $R\left\{ x_{1},\text{ }%
x_{2},..\right\} $ be the ring of formal power series in the variables $x_{1}$, $x_{2}$,
... with coefficients in $R$.

\begin{definition} \label{formalgrouplaw}
A commutative one--dimensional formal group law
over $R$ is a formal power series $\Psi \left( x,y\right) \in R\left\{
x,y\right\} $ such that \cite{Boch}
\begin{equation*}
1)\qquad \Psi \left( x,0\right) =\Psi \left( 0,x\right) =x,
\end{equation*}%
\begin{equation*}
2)\qquad \Psi \left( \Psi \left( x,y\right) ,z\right) =\Psi \left( x,\Psi
\left( y,z\right) \right) \text{.}
\end{equation*}
When $\Psi \left( x,y\right) =\Psi \left( y,x\right) $, the formal group law is
said to be commutative.
\end{definition}
The existence of an inverse formal series $\varphi
\left( x\right) $ $\in R\left\{ x\right\} $ such that $\Psi \left( x,\varphi
\left( x\right) \right) =0$ is a consequence of Definition \ref{formalgrouplaw}. Let  $B=\mathbb{Z}[b_{1},b_{2},...]$ be the  ring of integral polynomials in infinitely many variables. We shall consider the series $
F\left( s\right) = \sum_{i=0}^{\infty} b_i \frac{s^{i+1}}{i+1}$,
with $b_0=1$. Let $G\left( t\right)$ be its compositional inverse:
\beq \label{I.2}
G\left( t\right) =\sum_{k=0}^{\infty} a_k \frac{t^{k+1}}{k+1},
\eeq
 i.e. $F\left( G\left( t\right) \right) =t$. From this property, we deduce $a_{0}=1, a_{1}=-b_1, a_2= \frac{3}{2} b_1^2 -b_2,\ldots$.
The Lazard formal group law \cite{Haze} is defined by the formal power series
\[
\Psi_{\mathcal{L}} \left( s_{1},s_{2}\right) =G\left( G^{-1}\left(
s_{1}\right) +G^{-1}\left( s_{2}\right) \right).
\]
The coefficients of the power series $G\left( G^{-1}\left( s_{1}\right) +G^{-1}\left(
s_{2}\right) \right)$ lie in the ring $B \otimes \mathbb{Q}$ and generate over $\mathbb{Z}$ a subring $A \subset B \otimes \mathbb{Q}$, called the Lazard ring $L$.

For any commutative one-dimensional formal group law over any ring $R$, there exists a unique homomorphism $L\to R$ under which the Lazard group law is mapped into the given group law (the \textit{universal property} of the Lazard group).

Let $R$ be a ring with no torsion. Then, for any commutative one-dimensional formal group law $\Psi(x,y)$ over $R$, there exists a series $\psi(x)\in R[[x]] \otimes \mathbb{Q}$ such that
\[
\psi(x)= x+ O(x^2), \quad \text{and} \quad \Psi(x,y)= \psi^{-1}\left(\psi(x)+\psi(y)\right)\in R[[x,y]]\otimes \mathbb{Q}.
\]

The universal formal group plays the role of the general composition law admissible for the construction of the entropies of the $Z$-family.

See also \cite{Nov} and \cite{Quillen} for applications of formal groups in cobordism theory and \cite{PT2007CR}, \cite{PT2010ASN} and \cite{PT2015TRAN} for applications in number theory.

We also mention that in \cite{CT2019} the notion of \textit{formal rings} has been recently introduced as a natural extension of the notion of formal groups.

\section{Generalized logarithms and exponentials from group laws}
There is a simple construction allowing to define a generalized logarithm from a given group law.
\begin{definition} \label{defglog}
Let $G$ be a series of the form \eqref{I.2}. A  generalized group logarithm is a continuous, strictly concave, monotonically increasing function $\ln_{G}: (0,\infty)\to \mathbb{R}$, possibly depending on a set of real parameters, such that $\ln_{G}(\cdot)$ solves the functional equation for the group law corresponding to $G$, i.e
\beq
\ln_{G}(xy)= \Psi(\ln_{G}(x),\ln_{G}(y)) \label{glog}
\eeq
where
\beq
\Psi(x,y)=G(G^{-1}(x)+G^{-1}(y)). \label{GPsi}
\eeq
\end{definition}


It is easy to observe that the function $F_{G}(x)$ defined by
\beq
F_{G}(x):= G\left(\ln x\right)\label{Glog}
\eeq
satisfies  eq. \eqref{glog}, where $\Psi(x,y)$ is the group law \eqref{GPsi}.
Indeed, tt suffices to observe that
\bea \label{proofTh1}
F_{G}(xy)&=& G\left(\ln x+\ln y\right)=G\left(G^{-1}(F_{G}(x))+G^{-1}(F_{G}(y))\right)\\
\nn &=& \Psi(F_{G}(x),F_{G}(y)).
\eea 

The theorem can also be formulated in a field of characteristic zero in the class of formal power series. As is well known \cite{Haze}, for a 1-dimensional formal group law $\Psi(x,y)$ over a  torsion-free ring, there exists a formal series $G(t)$ of the form \eqref{I.2} that realizes eq. \eqref{GPsi}. Then the same relations \eqref{proofTh1} still hold.

\noi By way of an example, when $\Psi(x,y)=x+y$, we have directly that $G(t)=t$ and $\ln_{G}(x)=\ln x$. If
$\Psi(x,y)=x+y+(1-q)xy$, an associated function $G(t)$ is provided by $G(t)=\frac{e^{(1-q)t}-1}{1-q}$ and the group logarithm converts into the Tsallis logarithm
\beq
\ln_q(x):=\frac{x^{1-q}-1}{1-q}. \label{Tslog}
\eeq
\begin{remark}
The requirement of concavity of $\ln_{G}(x)$ is guaranteed by the condition
\beq
 a_{k} >    (k+1)  a_{k+1} \qquad \forall k\in\mathbb{N} \qquad \text{with } \{a_{k}\}_{k\in\mathbb{N}}\geq 0 \label{conc},
\eeq
which is also sufficient to ensure that the series $G(t)$ is absolutely and uniformly convergent with a radius $r=\infty$.
\end{remark}

In other words, given a group law, under mild hypotheses we may determine a generalized group logarithm by means of relation \eqref{Glog} and the condition \eqref{conc} (see \cite{PT2011PRE} for a construction of group logarithms from difference operators via the associated group exponential $G$).

\begin{definition} \label{defgexp}
The inverse of a generalized group logarithm will be called the associated generalized group exponential; it is defined by
\beq
exp_{G}(x)= e^{G^{-1}(x)}. \label{Gexp}
\eeq
\end{definition}
When $G(t)=t$, we have back the standard exponential; when $G(t)=\frac{e^{(1-q)t}-1}{1-q}$, we recover the $q$-exponential $e_{q}(x)=\left[1+(1-q)t\right]^{\frac{1}{1-q}}$, and so on.
\begin{remark}
From a computational point of view, observe that the formal compositional inverse $G^{-1}(s)$, such that $G(G^{-1}(s))=s$ and $G^{-1}(G(t))=t$ can be obtained by means of the Lagrange inversion theorem. We get the formal power series
\beq
G^{-1}(s)=s-\frac{a_1}{2}s^2+ \ldots
\eeq
By imposing the relation $\Psi(x,y)= G\left(G^{-1}(x)+G^{-1}(y)\right)$, with $\Psi(x,y)=x+y+ \text{higher order terms}$, we get a system of equations that allows to reconstruct the sequence $\{a_{k}\}_{k\in\mathbb{N}}$. Each element $a_{k}$  a priori depends on the set of parameters appearing in $\Psi(x,y)$.
\end{remark}

\section*{Acknowledgement}
This work has been partly supported by the research project
FIS2015-63966, MINECO, Spain, and by the ICMAT Severo Ochoa project
SEV-2015-0554 (MINECO). P.T. is member of the Gruppo Nazionale di Fisica Matematica (INDAM), Italy.

\section*{Author contribution}
P.T. and H.J.J developed the proposed research, wrote the paper and reviewed it in a happy and constructive collaboration.

\section*{Competing interests}
The authors declare no competing interests of any form.


\end{document}